\documentclass[11pt]{article}
\usepackage{epsfig}
\usepackage{latexsym}
\usepackage{amsmath}
\usepackage{verbatim}
\usepackage{enumerate}
\usepackage{setspace}
\usepackage{comment}
\usepackage{anysize}
\usepackage{epic}
\usepackage{eepic}
\usepackage{xcolor}
\usepackage{booktabs}

\usepackage[all,color,frame,import]{xy}

\newtheorem{theorem}{Theorem}[section]
\newtheorem{lemma}{Lemma}[section]

\newcommand\QED{\ifhmode\allowbreak\else\nobreak\fi
\quad\nobreak$\Box$\medbreak}

\newcommand{\proofstart}{\par\noindent \emph{Proof:} }
\newcommand{\proofend}{\QED\par}
\newenvironment{proof}{\proofstart}{\proofend}

\usepackage[letterpaper,hmargin=1in,vmargin=1in]{geometry}

\long\gdef\boxit#1{\vspace{5mm}\begingroup\vbox{\hrule\hbox{\vrule\kern5pt
\vbox{\kern5pt#1\kern5pt}\kern0pt\vrule}\hrule}\endgroup}

\begin{document}

\title{Two Multivehicle Routing Problems with Unit-Time Windows}

\author{Greg N. Frederickson\thanks{Dept. of Computer Sciences, Purdue University, West Lafayette,
    IN 47907. {\tt gnf@cs.purdue.edu}} \and Barry Wittman\thanks{Dept. of Computer Science, Elizabethtown College, Elizabethtown,
    PA 17022. {\tt wittmanb@etown.edu}}
}

\maketitle

\begin{abstract}
Two multivehicle routing problems are considered in the framework that a visit to a location must take place during a specific time window in order to be counted and all time windows are the same length.  In the first problem, the goal is to visit as many locations as possible using a fixed number of vehicles.  In the second, the goal is to visit all locations using the smallest number of vehicles possible.
For the first problem, we present an approximation algorithm whose output path collects a reward within a constant factor of optimal for any fixed number of vehicles.  For the second problem, our algorithm finds a 6-approximation to the problem on a tree metric, whenever a single vehicle could visit all locations during their time windows.\\

\end{abstract}

\section{Introduction}

In \cite{Frederickson3} we introduced polynomial time, constant-factor approximation algorithms to the Traveling Repairman Problem with unit-time windows.  In that problem, a single agent traveling at a fixed speed on a weighted graph must visit as many locations as possible during their time windows.  Although this is a problem general enough to model many practical routing and path-planning problems, it only considers a single agent.  Many commercial industries have large fleets of vehicles whose routes must be coordinated together for maximum efficiency.

In this paper we introduce approximation algorithms for some multivehicle problems, often called vehicle routing problems.  Our work in this area, naturally, focuses on the addition of time windows to vehicle routing problems without time constraints.  The vehicle routing problem was first introduced in \cite{Dantzig2}, which frames the problem with a central depot, a set of a trucks, and a set of locations requiring product from the depot.  The goal of this problem is to minimize the total mileage traveled by the fleet of trucks in servicing all locations.  Our version of the problem assumes that no time is required at the location to perform the service because service times can easily be absorbed into the structure of the graph \cite{Frederickson6}.

As with single vehicle problems, the addition of time constraints introduces several different kinds of optimization because it may no longer be possible to service all locations with a given fleet of vehicles and hard time constraints.  In \cite{Fakcharoenphol}, the authors find a constant approximation for the problem in which there are no hard time constraints, but the goal is to minimize the average time customers have to wait.  In \cite{Karuno2}, only a release time is given for each location, and the goal is to minimize the maximum lateness that any location is serviced after its release time.  For this problem, a PTAS is found, but only when the locations are on a path.  The authors of \cite{Hashimoto} similarly ``soften'' time window requirements by changing lateness into a cost function to be minimized.  The goal of the orienteering problem is to visit as many locations as possible before a global deadline.  In \cite{Blum3}, an algorithm is given that uses an $\alpha$-approximation to orienteering to find an $(\alpha + 1)$-approximation for the multiple-path orienteering problem, for which $k$ different vehicles are trying to maximize the total number of locations visited at least once before a global deadline.

With the introduction of multiple vehicles, it is reasonable to try to minimize the number of vehicles used.  In \cite{Bazgan}, the goal is to minimize the total number of vehicles needed to service all locations, starting and returning to a depot, with a hard constraint $D$ on the total distance a single vehicle can travel.  This distance limitation can be viewed as a global deadline.  A similar problem, with the requirement of returning to the depot removed, is discussed in \cite{Nagarajan} where a 4-approximation is given for tree metrics and an $O(\log D)$-approximation is given for general metrics.

In Section~\ref{section:trimming}, we define the Traveling Repairman Problem from \cite{Frederickson3}.  We will use approximation algorithms for this problem as subroutines.  More importantly, we will explain the concept of \emph{trimming} that plays an integral role in both the algorithms in \cite{Frederickson3} and in this paper.

In Section~\ref{section:multivehicle}, for any fixed $k$ we give a constant-factor polynomial-time approximation algorithm for the unrooted $k$-vehicle routing problem with unit-time windows.  This problem is \emph{unrooted} in the sense that vehicles are permitted to start at any time at any location and need not return to a central depot by a deadline.

In Section~\ref{section:minimumvehicle}, we consider a second problem, namely the Minimum Vehicle $OPT = 1$ Problem discussed in \cite{Nagarajan}.  In this problem, it is assumed that a single vehicle is enough to service all locations before a global deadline.  Our approximation algorithm finds paths that guarantee that some small number of vehicles is sufficient to service every location before the deadline.  A 14-approximation for this problem is given in \cite{Nagarajan}.  We change this problem from a single deadline to time windows in the unrooted setting and give a 6-approximation to this problem on tree metrics.  Most notable in this algorithm is the replacement of trimming by expansion, extending time windows to be longer rather than shorter as with the usual compression.

In these sections, we present approximation algorithms for two fundamental multivehicle problems.  Both the form of the problems and our approaches to solving them offer interesting contrasts.  While the multivehicle routing problem uses a simple algorithm that repeatedly runs a single vehicle repairman approximation, the analysis needed to show the performance of this algorithm requires a judicious choice of bounding terms.  On the other hand, the minimum vehicle approximation algorithm we present depends on a clever reversal of the trimming techniques from \cite{Frederickson3}, but the analysis is immediate.

\section{Repairman and Trimming}
\label{section:trimming}

We will use approximation algorithms for the Traveling Repairman Problem \cite{Frederickson3}, a 1-vehicle version of the problems considered in this paper.  In this problem, a repairman is presented with a set of \emph{service requests}.  Each service request is located at a node in a weighted, undirected graph and is assigned a \emph{time window} during which it is valid.  The goal of the repairman is to plan a route called a \emph{service run} that services as many requests as possible during their time windows while traveling at a given fixed speed.  Note that we consider the unrooted version of the problems, in which the agent may start at any time from any location and stop similarly.

We give approximation algorithms in \cite{Frederickson3} that guarantee a $3\gamma$-approximation to the Traveling Repairman Problem with unit-time windows and run in $\Gamma(n)$ time.  For a tree, $\gamma = 1$ and $\Gamma(n)$ is $O(n^4)$.  For a metric graph, $\gamma = 2 + \epsilon$ and $\Gamma(n)$ is $O(n^{O(1/\epsilon^2)})$, with the addition of improvements to \cite{Frederickson3} given in \cite{Chekuri2}.  A more thorough explanation of these improvements is available in \cite{Frederickson6}.

To achieve these results, we use a technique called trimming 
that is effective when we deal with unit-time windows.
Starting with time 0, we make divisions in time at values 
which are integer multiples of one half, i.e., 0, .5, 1, and so on.  
We assume that no request window starts on such a division,
because we can always redefine times to be decreased by a negligible amount.
We thus assume that the starting time for any window is positive.
Let a \emph{period} 
be the time interval from one division up to but not including the next division.
Because every service request has a time window exactly one unit long, half of that time window will be wholly contained within just one period, with the rest of the time window
divided between the preceding and following periods. 
We then trim each service request window to coincide precisely with the period wholly contained in it,
ignoring those portions of the request window that fall outside of the chosen period.

For the Repairman Problem,
the trimming may well lower the profit of the best service run, but by no more than a factor of 3, as given in the Limited Loss Theorem of \cite{Frederickson3}.

\section{Multivehicle Problem with Unit-Time Windows}
\label{section:multivehicle}

In this section we define the \emph{Multivehicle Routing Problem} and give a constant approximation to the unit-time window case when there are two vehicles.  We show how this approach can be extended to $k$ vehicles at the price of sacrificing a constant approximation.

We define the Multivehicle Routing Problem to take the same input as the Traveling Repairman Problem, but with a whole number specifying the number of vehicles.  The goal is to assign a service run for each vehicle such that the total profit is maximized.  As before, a service request can only be serviced once, and it must be serviced during its specified time window.  We assume that the time windows are all of unit length.

For the case of two vehicles, we run the algorithm called 2VEHICLE:
First divide the graph into periods of .5 time units and trim requests into those periods as explained in Section\ \ref{section:trimming}.  Next, run our single vehicle repairman approximation from \cite{Frederickson3} on the trimmed requests and call the resulting service run ${R_1}$.  Then, remove all the requests serviced by the first pass of our single vehicle approximation, run the approximation a second time, and call the resulting service run ${R_2}$.  Runs ${R_1}$ and ${R_2}$ are the output of 2VEHICLE.  Let $p(R)$ give the \emph{profit} collected by run $R$ for servicing locations, and let $c(R)$ give the \emph{cost} of traveling along run $R$ under metric $d$.

Let service runs $R_1^*$ and $R_2^*$ service disjoint subsets of requests and be optimal in the sense that the quantity $p(R_1^*) + p(R_2^*)$ is maximized.  W.l.o.g., let $p(R_1^*) \geq p(R_2^*)$.  Service runs ${R_1}$ and ${R_2}$ may overlap arbitrarily with service runs $R_1^*$ and $R_2^*$.  Let $p_1^*({R_1})$ be the profits earned by ${R_1}$ by servicing requests which were serviced by run $R_1^*$.  Similarly, let $p_2^*({R_1})$ be the profits earned by ${R_1}$ by servicing requests which were serviced by run $R_2^*$.  Define $p_1^*({R_2})$ and $p_2^*({R_2})$ in a similar manner.

\begin{lemma}
Algorithm 2VEHICLE gives a ${36\gamma^2 \over 12\gamma - 1}$-approximation to the Multivehicle Routing Problem with unit-time windows for 2 vehicles on any metric graph.
\end{lemma}

\begin{proof}
The profit collected by run $R_1^*$ cannot be greater than the profit collected by a single vehicle optimal tour.  Thus, by our earlier arguments about the impact of trimming, $p({R_1}) \geq {1 \over 3\gamma}p(R_1^*)$.  By definition of $p_1^*({R_1})$ and $p_2^*({R_1})$, $p({R_1}) \geq p_1^*({R_1}) + p_2^*({R_1})$.

Again because of trimming, the profit collected by run ${R_2}$ is no smaller than ${1 \over 3\gamma}$ of the larger profit left either in runs $R_1^*$ or $R_2^*$ after requests serviced by run ${R_1}$ have been removed.  Thus,
{\footnotesize
$$p({R_2}) \geq \max\left\{{1\over 3\gamma}\Big(p(R_1^*) - p_1^*({R_1})\Big), {1\over 3\gamma}\Big(p(R_2^*) - p_2^*({R_1})\Big)\right\}$$}

As a consequence, the total value of profit collected by ${R_1}$ and ${R_2}$ is:
{\footnotesize
\begin{eqnarray*}
p({R_1}) + p({R_2}) &\geq& p({R_1}) + \max\left\{{1\over 3\gamma}\Big(p(R_1^*) - p_1^*({R_1})\Big), {1\over 3\gamma}\Big(p(R_2^*) - p_2^*({R_1})\Big)\right\}\\
&\geq& p({R_1}) + {1\over 2}\left({1\over 3\gamma}\Big(p(R_1^*) - p_1^*({R_1})\Big) + {1\over 3\gamma}\Big(p(R_2^*) - p_2^*({R_1})\Big)\right)\\
&=& p({R_1}) + {1\over 6\gamma}p(R_1^*) - {1 \over 6\gamma}p_1^*({R_1}) + {1\over 6\gamma}p(R_2^*) - {1\over 6\gamma}p_2^*({R_1})\\
 &\geq& p({R_1}) + {1\over 6\gamma}\Big(p(R_1^*) + p({R_2^*})\Big) -  {1\over 6\gamma}\Big(p_1(R_1) + p_2({R_1})\Big)\\
&\geq& {6\gamma - 1 \over 6\gamma}p({R_1}) + {1\over 6\gamma}\Big(p(R_1^*) + p(R_2^*)\Big)\\
&\geq& {6\gamma - 1 \over 18\gamma^2}p(R_1^*) + {1\over 6\gamma}\Big(p(R_1^*) + p(R_2^*)\Big)\\
&\geq& {6\gamma - 1 \over 36\gamma^2}(p(R_1^*) + p(R_2^*))+ {1\over 6\gamma}\Big(p(R_1^*) + p(R_2^*)\Big)\\
&\geq& {12\gamma - 1 \over 36\gamma^2}\Big(p(R_1^*) + p(R_2^*)\Big)
\end{eqnarray*}}\end{proof}

For the 2-vehicle problem on a tree, where $\gamma = 1$, algorithm 2VEHICLE obtains at least ${11 \over 36}$ of total possible profit.  Note that this greedy multi-pass algorithm is similar to the one used for multiple-path orienteering in \cite{Blum3}.  Using an analysis technique from there, we could achieve a $(3\gamma + 1)$-approximation to the Multivehicle Routing Problem.  However, with the more advanced analysis given, we achieve a better approximation for 2 as well as other small numbers of vehicles.

Our technique can be extended to 3 vehicles.  For a tree, the approximation ratio is $243/71$, as shown below.
{\footnotesize
\begin{eqnarray*}
p({R_1}) + p({R_2}) + p({R_3}) &\geq& p({R_1}) + p({R_2})+ \max\left\{{1\over 3}\Big(p(R_1^*) - p_1^*({R_1}) - p_1^*({R_2})\Big),\right.\\
& & \left.{1\over 3}\Big(p(R_2^*) - p_2^*({R_1}) - p_2^*({R_2})\Big), {1\over 3}\Big(p(R_3^*) - p_3^*({R_1}) - p_3^*({R_2})\Big)\right\}\\
&\geq& p({R_1}) + p({R_2}) + {1\over 9}\Big(p(R_1^*) - p_1^*({R_1}) - p_1^*({R_2})\Big) +\\ 
& &{1\over 9}\Big(p(R_2^*) - p_2^*({R_1}) - p_2^*({R_2})\Big) + {1\over 9}\Big(p(R_3^*) - p_3^*({R_1}) - p_3^*({R_2}) \Big)\\
&=& p({R_1}) + p({R_2}) + {1 \over 9}\Big(p(R_1^*) + p(R_2^*) + p(R_3^*)\Big)\\
& & -~{1 \over 9}\Big(p_1^*(R_1) + p_2^*(R_1) + p_3^*(R_1)\Big) - {1 \over 9}\Big(p_1^*(R_2) + p_2^*(R_2) + p_3^*(R_2)\Big)\\
&\geq& {8 \over 9}p({R_1}) + {8 \over 9}p({R_2})  + {1\over 9}\Big( p(R_1^*) + p(R_2^*) + p(R_3^*) \Big)\\
&\geq& {8 \over 9}\left({11 \over 36}(p(R_1^*) + p(R_2^*))\right) + {1\over 9}\Big( p(R_1^*) + p(R_2^*) + p(R_3^*) \Big)\\
&\geq& {8 \over 9}\left({22 \over 108}(p(R_1^*) + p(R_2^*) + p(R_3^*))\right) + {1\over 9}\Big( p(R_1^*) + p(R_2^*) + p(R_3^*) \Big)\\
&\geq& {71 \over 243}\Big(p(R_1^*) + p(R_2^*) + p(R_3^*)\Big)
\end{eqnarray*}}
Our technique logically extends to $k$ vehicles on a tree.  Let $R_1^*, R_2^*, \ldots R_k^*$ be the $k$ disjoint optimal runs in decreasing order of profit.  Let $R_1, R_2, \ldots R_k$ be the $k$ runs produced by our algorithm. Using the same structure used for $k = 2$ and $k = 3$,
{\footnotesize
\begin{eqnarray*}
\sum_{i = 1}^k p(R_i) &\geq& \sum_{i = 1}^{k - 1} p(R_i) + \max_{1 \leq i \leq k} \left\{{1\over 3}\left(p(R_i^*) - \sum_{j = 1}^{k - 1} p_i^*(R_j)\right) \right\}\\
&\geq& {3k - 1 \over 3k}\sum_{i = 1}^{k - 1} p(R_i) + {1 \over 3k} \sum_{i = 1}^k p(R_i^*)
\end{eqnarray*}}
We can substitute the approximation for the first $k - 1$ runs into the first term of the bound given above.  Then, that term will still only give an approximation as a ratio of $\sum_{i = 1}^{k - 1} p(R_i^*)$.  Because we know that the profit of $R_k^*$ is no greater than the profit of any of the runs in the sum, we can bound the approximation in terms of a ratio of $\sum_{i = 1}^{k} p(R_i^*)$ with an additional factor of $(k - 1)/k$.  Thus, for trees we find at least a $P(k)$ fraction of optimal profit, where $P(k)$ is defined as follows:
{\footnotesize
\begin{eqnarray*}
P(1) &\geq& {1 \over 3}\\
P(k) &\geq& \left( {3k^2 - 4k + 1 \over 3k^2} \right)P(k  - 1) + {1 \over 3k}
\end{eqnarray*}}
For general metric graphs with a $3\gamma$-approximation to the Traveling Repairman Problem, we can generalize the analysis and the recurrence given to include a factor of $\gamma$.  Very similar analysis shows that we find a related $P_{\gamma}(k)$ fraction of profit, defining $P_{\gamma}(k)$ as follows:
{\footnotesize
\begin{eqnarray*}
P_{\gamma}(1) &\geq& {1 \over 3\gamma}\\
P_{\gamma}(k) &\geq& \left( {3\gamma k^2 - (3\gamma + 1) k + 1 \over 3\gamma k^2} \right)P_{\gamma}(k  - 1) + {1 \over 3\gamma k}
\end{eqnarray*}}
In Table \ref{table:P(k) values} we list values of $P_{\gamma}(k)$ for both tree and graph versions for several values of $k$.  Note that the quality of the approximation degrades very slowly.  A simple proof by induction establishes that our approximation factor is always better than the $(3\gamma + 1)$-approximation we could have achieved using the style of analysis from \cite{Blum3}.
\begin{table}[!hbt]
\centering
\begin{tabular}{ c c c c c c c c}
\toprule
$k$  & & $\gamma = 1$  &  & & \multicolumn{2}{c}{$\gamma = 2 + \epsilon$ when $\epsilon = 1$} \smallskip\\
\midrule
$2$ & & ${11\over 36}$ & $\approx 0.3056$  & & ${35\over 324}$ & $\approx 0.1080$\smallskip\\
$3$ & & ${71\over 243}$ & $\approx 0.2922$ & & ${698\over 6561}$ & $\approx 0.1064$\smallskip\\
$4$ & & ${1105\over 3888}$ & $\approx 0.2842$ & & ${16589\over 157464}$ & $\approx 0.1054$\smallskip\\
$8$ & & ${31871\over 118098}$ & $\approx 0.2699$ & & ${1601600479\over 15496819560}$ & $\approx 0.1034$\smallskip\\
$16$ & & ${524344607599\over 2008387814976}$ & $\approx 0.2611$ & & ${2755777442862301661\over 27017034353459841780}$ & $\approx 0.1020$\\
\bottomrule 
\end{tabular}
\caption{Values of $P_{\gamma}(k)$ for selected $k$ and $\gamma$.  Although almost 50\% better performance can be gained with small $\epsilon$, we let
 $\epsilon = 1$ for simplicity of presentation.}
\label{table:P(k) values}
\end{table}


\section{Vehicle Routing on a Tree When a Single Vehicle Suffices}
\label{section:minimumvehicle}

We now consider an approximation algorithm for a special case of a minimum vehicle routing problem with unit-time windows.  We call this problem the \emph{Minimum Vehicle $OPT = 1$ Problem}.  The input to this problem is identical to the Traveling Repairman Problem defined in Section~\ref{section:trimming}.  We also assume that only a single repairman is needed to service all requests.  Below we define an algorithm called SINGLE-REPAIR that finds no more than 6 independent runs that will service all requests, for the case of tree metrics.  Should our algorithm not produce 6 runs that service all requests, then we will know that no single vehicle service tour exists.

We begin SINGLE-REPAIR by making divisions in time spaced one unit apart, starting at time 0.  Let a period be the time interval from one division up to but not including the next division.  Note that this step defines periods to be twice as long as that given in the algorithm described in Section~\ref{section:trimming}.  We number the periods in order of increasing time, starting with the number 0.  If needed, we perturb the time windows by some negligible amount so that no time window begins exactly on a division.  Thus, every time window will intersect exactly two periods.  We expand each time window so that it fills both of the periods it intersects, doubling its length.  We partition the time windows into two sets:  If the first of the two periods a time window fills is even, we put the window in set $\mathcal{E}$.  Otherwise, we put the window in set $\mathcal{O}$.

We run the repairman algorithm on trimmed windows from \cite{Frederickson3} on set $\mathcal{E}$ to find a shortest run $R_\mathcal{E}$ servicing the maximum number of requests in that set.  Then, we run the same algorithm on the windows in set $\mathcal{O}$ resulting in a run $R_\mathcal{O}$. This algorithms gives optimal runs for trimmed windows on a tree; thus, $R_\mathcal{E}$ and $R_\mathcal{O}$ are optimal runs on their respective sets.  Since we know that a single vehicle can service the requests of all of the unexpanded windows, $R_\mathcal{E}$ and $R_\mathcal{O}$ must service the requests of all of the windows in their sets.

In order to convert $R_\mathcal{E}$ and $R_\mathcal{O}$ into runs on unexpanded windows, we make two additional copies of each.  For both sets of runs, we move one copy back in time by 1 time unit and one copy forward in time by 1 time unit.

\begin{theorem}
Algorithm SINGLE-REPAIR finds a 6-approximation to the Minimum Vehicle $OPT = 1$ Problem when the underlying graph is a tree.
\end{theorem}

\begin{proof}
Let us consider the case for path $R_\mathcal{E}$.  If $R_\mathcal{E}$ services a request of an expanded time window at time $t$, then the corresponding unexpanded time window either contains $t$, precedes $t$, or follows $t$.  If the original time window contains $t$, then $R_\mathcal{E}$ services its request.  If the time window precedes $t$, then the copy of $R_\mathcal{E}$ starting 1 time unit earlier must service its request.  If the time window follows $t$, then the copy of $R_\mathcal{E}$ starting 1 time unit later must service its request.  Because the argument is identical for $R_\mathcal{O}$, our algorithm finds 6 paths which service all service requests.
\end{proof}

\singlespacing
\bibliography{bibliography}

\end{document}